\documentclass[aos, authoryear]{imsart}

\usepackage{amsmath}
\usepackage{amssymb}
\usepackage{amsthm}
\usepackage{enumitem}
\usepackage[totalwidth=450pt, totalheight=630pt]{geometry}
\usepackage{graphicx}
\usepackage{hyperref} 
\usepackage{mathrsfs}
\usepackage{natbib} 
\usepackage[linesnumbered, ruled]{algorithm2e} 
\usepackage{cleveref}
\usepackage{tikz}
\usetikzlibrary{calc}
\usepackage{verbatim}
\usepackage{xcolor}

\arxiv{arXiv:1610.03944}

\startlocaldefs

\newlist{inlinelist}{enumerate*}{1}
\setlist*[inlinelist]{label=(\arabic*)}

\SetAlgoCaptionSeparator{.}

\definecolor{mid_green}{HTML}{99d8c9}

\newtheorem{corollary}{Corollary}
\newtheorem{definition}{Definition}
\newtheorem{lemma}[corollary]{Lemma}
\newtheorem{theorem}[corollary]{Theorem}

\theoremstyle{definition}
\newtheorem{example}{Example}

\newtheoremstyle{custom}{}{}{\itshape}{}{\bfseries}{.}{.5em}{\thmnote{#3}}
\theoremstyle{custom}
\newtheorem*{customthm}{Theorem}

\setlist{listparindent=\parindent, parsep=0pt}

\newcommand{\PP}{\mathbb{P}}
\newcommand{\RR}{\mathbb{R}}
\newcommand{\ZZ}{\mathbb{Z}}
\DeclareMathOperator*{\argmax}{arg\,max}

\endlocaldefs

\begin{document}

\begin{frontmatter}

\title{Rank Verification for Exponential Families}
\runtitle{Rank Verification for Exp. Families}

\begin{aug}
  \author{\fnms{Kenneth} \snm{Hung}\corref{}\ead[label=e1]{kenhung@berkeley.edu}}
  \and
  \author{\fnms{William} \snm{Fithian}\ead[label=e2]{wfithian@berkeley.edu}}

  \runauthor{K. Hung and W. Fithian}

  \affiliation{University of California, Berkeley}

  \address{Department of Mathematics\\ 951 Evans Hall, Suite 3840\\ Berkeley, CA 94720-3840\\
  \printead{e1}}

  \address{Department of Statistics\\ 301 Evans Hall\\ Berkeley, CA 94720\\
          \printead{e2}}

\end{aug}

\begin{abstract}
Many statistical experiments involve comparing multiple population groups. For example, a public opinion poll may ask which of several political candidates commands the most support; a social scientific survey may report the most common of several responses to a question; or, a clinical trial may compare binary patient outcomes under several treatment conditions to determine the most effective treatment. Having observed the ``winner'' (largest observed response) in a noisy experiment, it is natural to ask whether that candidate, survey response, or treatment is actually the ``best'' (stochastically largest response). This article concerns the problem of {\em rank verification} --- post hoc significance tests of whether the orderings discovered in the data reflect the population ranks. For exponential family models, we show under mild conditions that an unadjusted two-tailed pairwise test comparing the first two order statistics (i.e., comparing the ``winner'' to the ``runner-up'') is a valid test of whether the winner is truly the best. We extend our analysis to provide equally simple procedures to obtain lower confidence bounds on the gap between the winning population and the others, and to verify ranks beyond the first.
\end{abstract}



\end{frontmatter}

\section{Introduction}
\label{sec:intro}

\subsection{Motivating Example: Iowa Republican Caucus Poll}
\label{sec:iowa}

\Cref{tbl:poll} shows the result of a Quinnipiac University poll asking 890 Iowa Republicans their preferred candidate for the Republican presidential nomination \citep{quinnipiac}. Donald Trump led with $31\%$ of the vote, Ted Cruz came second with $24\%$, Marco Rubio third with $17\%$, and ten other candidates including ``Don't know'' trailed behind. 

\begin{table}[htbp]
\centering
\begin{tabular}{c c c c}
	\hline
	Rank & Candidate & Result & Votes \\
	\hline
	$1$ * & Trump & $31\%$ & $276$ \\
	$2$ * & Cruz & $24\%$ & $214$ \\
	$3$ * & Rubio & $17\%$ & $151$ \\
	$4$ * & Carson & $8\%$ & $71$ \\
	$5$ & Paul & $4\%$ & $36$ \\
	$6$ & Bush & $4\%$ & $36$ \\
	$7$ & Huckabee & $3\%$ & $27$ \\
	$\vdots$ & $\vdots$ & $\vdots $ & $\vdots$ \\
	\hline
\end{tabular}
\caption{Results from a February 1, 2016 Quinnipiac University poll of $890$ Iowa Republicans. To compute the last column (Votes), we make the simplifying assumption that the reported percentages in the third column (Result) are raw vote shares among survey respondents. The asterisks indicate that the rank is verified at level $0.05$ by a stepwise procedure.}
\label{tbl:poll}
\end{table}

Seeing that Trump leads this poll, several salient questions may occur to us: Is Trump really winning, and if so by how much? Furthermore, is Cruz really in second, is Rubio really in third, and so on? Note that there is implicitly a problem of multiple comparisons here, because if Cruz had led the poll instead, we would be asking a different set of questions (``Is Cruz really winning,'' etc.). Indeed, the selection issue appears especially pernicious due to the so-called ``winner's curse'': given that Trump leads the poll, it more likely than not overestimates his support.

Nevertheless, if we blithely ignore the selection issue, we might carry out the following analyses to answer the questions we posed before at significance level $\alpha = 0.05$. We assume for simplicity that the poll represents a simple random sample of Iowa Republicans; i.e., that the data are a multinomial sample of size $890$ and underlying probabilities $\left(\pi_{\text{Trump}}, \pi_{\text{Cruz}}, \ldots\right)$. (The reality is a bit more complicated: before releasing the data, Quinnipiac has post-processed it to make the reported result more representative of likely caucus-goers. The raw data is proprietary.)

\begin{enumerate}
\item {\em Is Trump really winning?} If Trump and Cruz were in fact tied, then Trump's share of their combined 490 votes would be distributed as $\text{Binomial}\left(490, 0.5\right)$. Because the (two-tailed) $p$-value for this pairwise test is $p = 0.006$, we reject the null and conclude that Trump is really winning.

\item {\em By how much?} Using an exact $95\%$ interval for the same binomial model, we conclude Trump has at least $7.5\%$ more support than Cruz (i.e., $\pi_{\text{Trump}} \ge 1.075 \,\pi_{\text{Cruz}}$) and also leads the other candidates by at least as much.

\item {\em Is Cruz in second, Rubio in third, etc.?} We can next compare Cruz to Rubio just as we compared Trump to Cruz (again rejecting because $214$ is significantly more than half of 365), then Rubio to Carson, and so on, continuing until we fail to reject. The first four comparisons are all significant at level $0.05$, but Paul and Bush are tied so we stop.
\end{enumerate}

Perhaps surprisingly, all of the three procedures described above are statistically valid despite their ostensibly ignoring the implicit multiple-comparisons issue. In other words, Procedures 1 and 2 control the Type I error rate at level $\alpha$ and Procedure 3 controls the familywise error rate (FWER) at level $\alpha$. The remainder of this article is devoted to justifying these procedures for the multinomial family, and extending to analogous procedures in other exponential family settings. While methods analogous to Procedures 1 and 2 have been justified previously for balanced independent samples from log-concave location families \citep{Gutmann:1987fk,Stefansson:1988wj}, they have not been justified in exponential families before now.

\subsection{Generic Problem Setting and Main Result}
\label{sec:probset}

Generically, we will consider data drawn from an exponential family model with density
\begin{equation}
\label{eq:expfam}
X \sim \exp\left(\theta'x - \psi\left(\theta\right)\right) g\left(x\right),
\end{equation}
with respect to either the Lebesgue measure on $\RR^n$ or counting measure on $\ZZ^n$. We assume further that $g\left(x\right)$ is symmetric with respect to permutation, and Schur concave, a mild technical condition defined in \Cref{sec:maj}. In addition to the multinomial family, model~\eqref{eq:expfam} also encompasses settings such as comparing independent binomial treatment outcomes in a clinical trial, competing sports teams under a Bradley--Terry model, entries of a Dirichlet distribution, and many more; see \Cref{sec:maj} for these and other examples.

We will generically use the term {\em population} to refer to the treatment group, sports team, political candidate, etc.\ represented by a given random variable $X_j$. As we will see, $\theta_j \ge \theta_k$ if and only if $X_j$ is stochastically larger than $X_k$; thus, there is a well-defined stochastic ordering of the populations that matches the ordering of the entries of $\theta$. We will refer to the population with maximal $\theta_j$ as the {\em best}, the population with second largest $\theta_j$ as the {\em second best}, the one with maximal $X_j$ as the {\em winner}, and the one with the second-largest $X_j$ as the {\em runner-up}, where ties between observations are broken randomly to obtain a full ordering. Following the convention in the ranking and selection literature, we assume that if there are multiple largest $\theta_j$, then one is arbitrarily marked as the best. Note that in cases where it is more interesting to ask which is the smallest population (for example, if $X_j$ is the number of patients on treatment $j$ who suffer a heart attack during a trial) we can change the variables to $-X$ and the parameters to $-\theta$; this does not affect the Schur concavity assumption. 

Write the order statistics of $X$ as
$$X_{[1]} \ge X_{[2]} \ge \cdots \ge X_{[n]},$$
where $[j]$ will denote the random index for the $j$-th order statistic. Thus, $\theta_{[j]}$ is the entry of $\theta$ corresponding to the $j$-th order statistic of $X$ (so $\theta_{[1]}$ might {\em not} equal $\max_j \theta_j$, for example).

In each of the above examples, there is a natural exact test we could apply to test $\theta_j=\theta_k$ for any two {\em fixed} populations $j$ and $k$. In the multinomial case, we would apply the conditional binomial test based on the combined total $X_j+X_k$ as discussed in the previous section. For the case of independent binomials we would apply Fisher's exact test, again conditioning on $X_j+X_k$. These are both examples of a generic UMPU pairwise test in which we condition on the other $n-2$ indices (notated $X_{\setminus \{j,k\}}$) and $X_j+X_k$, and reject the null if $X_j$ is outside the $\alpha/2$ and $1-\alpha/2$ quantiles of the conditional law $\mathcal{L}_{\theta_j=\theta_k}(X_j \mid X_j+X_k, X_{\setminus\{j,k\}})$. Crucially, this null distribution does not depend on the value of $\theta$ provided that $\theta_j=\theta_k$. We call this test the (two-tailed) {\em unadjusted pairwise test} since it makes no explicit adjustment for selection. Similarly, inverting this test for other values of $\theta_j-\theta_k$ yields an {\em unadjusted pairwise confidence interval}. (To avoid trivialities in the discrete case, we assume these procedures are appropriately randomized at the rejection thresholds to give exact level-$\alpha$ control.)

Generalizing the procedures described in \Cref{sec:iowa} we obtain the following:
\begin{enumerate}
\item {\em Is the winner really the best?} To test the hypothesis $H:\; \theta_{[1]} \leq \max_{j \neq [1]} \theta_j$: Carry out the unadjusted pairwise test comparing the winner to the runner-up. If the test rejects at level $\alpha$, reject $H$ and declare that the winner is really the best.
\item {\em By how much?} To construct a lower confidence bound for $\theta_{[1]} - \max_{j \neq [1]} \theta_j$: Construct the unadjusted pairwise confidence interval comparing the winner to the runner-up, and report the lower confidence bound obtained for $\theta_{[1]} - \theta_{[2]}$ if it is nonnegative, report $-\infty$ otherwise.
\item {\em Is the runner-up really the second best, etc.?} Continue by comparing the runner-up to the second runner-up, again using the unadjusted pairwise test, and so on down the list comparing adjacent values. Stop the first time the test does not reject; if there are $j$ rejections, declare that 
\[
\theta_{[1]} > \theta_{[2]} > \cdots > \theta_{[j]} > \max_{k > j} \theta_{[j]}
\]
\end{enumerate}
Procedures 2 and 3 are conservative stand-ins for exact, but slightly more involved, conditional inference procedures. In particular, as we will see, reporting $-\infty$ in Procedure 2 is typically much more conservative than is necessary.

We now state our main theorem: under a mild technical assumption, Procedures 1--3 described above are statistically valid, even accounting for the selection.
\begin{theorem}
\label{thm:main}
Assume the model~\eqref{eq:expfam} holds and $g\left(x\right)$ is a Schur-concave function. Then:
\begin{enumerate}
\item Procedure 1 has exact level $\alpha$ conditional on $H$ being true (conditional on the best population not winning), and marginally has level $\alpha \cdot \PP(H \text{ is true}) \leq \alpha \left(1-\frac{1}{n}\right)$.
\item Procedure 2 gives a conservative $1-\alpha$ lower confidence bound for $\theta_{[1]} - \max_{j \ne [1]} \theta_{j}$.
\item Procedure 3 is a conservative stepwise procedure with FWER no larger than $\alpha$.
\end{enumerate}
\end{theorem}

Note that Theorem~\ref{thm:main} implies that we could actually replace $\alpha$ with $\frac{n}{n-1} \alpha$ to obtain a more powerful version of Procedure 1 when $n$ is not too large.

We define Schur-concavity and discuss its properties in \Cref{sec:maj}. Because any log-concave and symmetric function is Schur-concave, \Cref{thm:main} applies to all of the cases discussed above. The proof combines the conditional selective-inference framework of \citet{Fithian:2014ws} with classical multiple-testing methods, as well as new technical tools involving majorization and Schur-concavity.

Note that these procedures make an implicit adjustment for selection because they use two-tailed, rather than one-tailed, unadjusted tests. If we instead based our tests on an independent realization $X^* = (X_1^*, \ldots, X_n^*)$ then, for example, Procedure 1 could use a right-tailed version of the unadjusted pairwise test. In the case $n = 2$, Procedure 1 amounts to a simple two-tailed test of the null hypothesis $\theta_1 = \theta_2$, and it is intuitively clear that a one-tailed test would be too liberal. More surprising is that, no matter how large $n$ is, Procedures 1--3 require no further adjustment beyond what is required when $n = 2$. 

\subsection{Related work}
\label{sec:related}

Rank verification has been studied extensively in the ranking and selection literature. See \citet{Gupta:1971wk, Gupta:1985bj} for surveys of the subset selection literature. The two main formulations of ranking and selection are closely related to procedures for multiple comparisons with the best treatment \citep{Edwards:1983eo,Hsu:1984gb},  but more powerful methods are available in some cases for procedures involving only the first sample rank, the problem of comparisons with the sample best; see \citet{Hsu:1996} for an overview and discussion of the relationships between these problems. 

Comparisons with the sample best have been especially well-studied and the validity of Procedures 1 and 2 have been established in a different setting: balanced independent samples from log-concave location families. \citet{Gutmann:1987fk} prove the validity of Procedure 1 in this setting, and \citet{Bofinger:1991hv,Maymin:1992fz, Karnnan:2009iv} give similar results for other models including scale and location-scale families. \citet{Stefansson:1988wj} provide an alternative proof for the validity of Procedure 1 in the same setting, leading to a lower confidence bound analogous to that of Procedure 2; interestingly, the proof involves a very early application of the partitioning principle, later developed into fundamental technique in multiple comparisons \citep{Finner:2002ju}. These results use very different technical tools than the ones we use here, require independence between the different groups (ruling out, for example, the multinomial family), and do not address the exponential family case. Because most exponential families are not location-scale families (the Gaussian being a notable exception), and because our results involve more general dependence structures, both our proof techniques and our technical results are complementary to the techniques and results in the above works.

For the multinomial case, \citet{Gupta:1967wg}, discussed in \Cref{sec:subsetsel}, remain the state of the art in finite-sample tests; \citet{Gupta:1976vd} discuss related approaches for Poisson models. \citet{Berger:1980ev} mentions an alternative, simpler rule which performs a binomial test on each population, but its power does not necessarily increase as the size $m$ of observations increases in cases like $\text{Multinomial}(m; 2/3, 1/3, 0, \ldots, 0)$. \citet{Nettleton:2009ht} proves validity for an asymptotic version of the winner-versus-runner-up test, and \citet{Gupta:1989fe} consider an empirical Bayes approach for selecting the best binomial population wherein a parametric prior distribution is assumed for the success probabilities for the different populations. \citet{Ng:2007cn} discuss an exact test for a modified problem in which the maximum count is fixed instead of the total count; that is, we sample until the leading candidate has at least $m$ votes. As \Cref{sec:subsetsel} shows, our test can be much more powerful than the one in \citet{Gupta:1967wg}, especially if there are many candidates, because of the way our critical rejection threshold for $X_{[1]} - X_{[2]}$ adapts to the data. Thus, our work closes a significant gap in the ranking and selection literature, extending the result of \citet{Gutmann:1987fk} and others to new families like the multinomial, independent binomials, and many others.

\subsection{Outline}

\Cref{sec:maj} defines Schur concavity, and gives several examples satisfying this condition. \Cref{sec:winnerbest} justifies Procedure 1 and compares its power to that of \citet{Gupta:1967wg}. \Cref{sec:confbound,sec:stepwise} justify Procedures 2 and 3 respectively, and \Cref{sec:disc} concludes.

\section{Majorization and Schur concavity}
\label{sec:maj}

\subsection{Definitions and basic properties}

We start by reviewing the notion of {\em majorization}, defined on both $\RR^n$ and $\ZZ^n$.

\begin{definition}
For two vectors $a$ and $b$ in $\RR^n$ (or $\ZZ^n$), suppose sorting the two vectors in descending order gives
$a_{\left(1\right)} \ge \cdots \ge a_{\left(n\right)}$ and $b_{\left(1\right)} \ge \cdots \ge b_{\left(n\right)}$. We say that $a \succeq b$ ($a$ majorizes $b$) if for $1 \le i < n$,
\begin{align*}
a_{\left(1\right)} + \cdots + a_{\left(i\right)} & \ge b_{\left(1\right)} + \cdots + b_{\left(i\right)}, \quad \text{and}\\
a_{\left(1\right)} + \cdots + a_{\left(n\right)} & = b_{\left(1\right)} + \cdots + b_{\left(n\right)}.
\end{align*}
This forms a partial order in $\RR^n$ (or $\ZZ^n$).
\end{definition}

Intuitively, majorization is a partial order that monitors the evenness of a vector: the more even a vector is, the ``smaller'' it is. There are two properties of majorization that we will use in the proofs.

\begin{lemma}\leavevmode
\label{lma:twoprop}
\begin{enumerate}
\item Suppose $\left(x_1, x_2, x_3, \ldots\right)$ and $\left(x_1, y_2, y_3, \ldots\right)$ are two vectors in $\RR^n$. Then
\[\left(x_1, x_2, x_3, \ldots\right) \succeq \left(x_1, y_2, y_3, \ldots\right) \text{ if and only if } \left(x_2, x_3, \ldots\right) \succeq \left(y_2, y_3, \ldots\right).\]
\item (Principle of transfer) If $x_1 > x_2$ and $t \ge 0$, then
\[\left(x_1 + t, x_2, x_3, \ldots\right) \succeq \left(x_1, x_2 + t, x_3, \ldots\right).\]
If $t \le 0$, the majorization is reversed.
\end{enumerate}
\end{lemma}

\begin{proof}\leavevmode
\begin{enumerate}
\item The property follows from an equivalent formulation of majorization listed in \citet{Marshall:2010hb}, where $x \succeq y$ if and only if
\[\sum_{j=1}^n x_n = \sum_{j=1}^n y_n \quad \text{ and } \quad \sum_{j=1}^n \left(x_j - a\right)_+ \ge \sum_{j=1}^n \left(y_j - a\right)_+ \text{ for all } a \in \RR.\]

\item Proved in \citet{Marshall:2010hb}. \qedhere
\end{enumerate}
\end{proof}

\begin{definition}
A function $g$ is Schur-concave if $x \succeq y$ implies $g\left(x\right) \le g\left(y\right)$.
\end{definition}

A Schur-concave function is symmetric by default since $a \succeq b$ and $b \succeq a$ if and only if $b$ is a permutation of the coordinates of $a$. Conversely a symmetric and log-concave function is Schur-concave \citep{Marshall:2010hb}. Interestingly, \citet{Gupta:1984fw} also show that, in the context of independent location families, Schur concavity of the probability density is equivalent to monotone likelihood ratio.

\subsection{Examples}

Many common exponential family models have Schur-concave carrier densities. Below we give a few examples:

\begin{example}[Independent binomial treatment outcomes in a clinical trial]
 If each of $n$ different treatments are applied to $m$ patients independently, the number of positive outcomes $X_j$ for treatment $j$ is $\text{Binomial}\left(m, p_j\right)$. The best treatment would be the treatment with the highest success probability $p_j$. The joint distribution of $X$ is given by
\[p\left(x\right) \propto \exp\left(\sum_j x_j \log\frac{p_j}{1-p_j}\right) \frac{1}{x_1! \left(m-x_1\right)! \cdots x_n! \left(m-x_n\right)!}\]
The carrier measure above is Schur-concave. The unadjusted pairwise test in this family is Fisher's exact test.
\end{example}

\begin{example}[Competitive sports under the Bradley--Terry model]
Suppose $n$ players compete in a round robin tournament, where player $j$ has ability $\theta_j$, and the probability of player $j$ winning against player $k$ is
\[\frac{e^{\theta_j - \theta_k}}{1 + e^{\theta_j - \theta_k}} = \frac{e^{\left(\theta_j - \theta_k\right) / 2}}{e^{\left(\theta_j - \theta_k\right) / 2} + e^{\left(\theta_k - \theta_j\right) / 2}}.\]

Let $Y_{jk}$ be an indicator for the match between player $j$ and $k$, where we take $Y_{jk} = 1$ if $j$ beats $k$ and $Y_{jk} = 0$ if $k$ beats $j$. For symmetry, we will also adopt the convention that $Y_{jk} + Y_{kj} = 1$. Thus the joint distribution of $Y=\left(Y_{jk}\right)_{j\neq k}$ is
\[p\left(y\right) \propto \exp\left(\sum_j 2\theta_j \sum_{k \ne j} y_{jk}\right)  = \exp\left(2\theta' x\right),\]
where $x_j = \sum_{k \ne j} y_{jk}$. In other words, if $X_j$ is the number of wins by player $j$, then $X=(X_1,\ldots,X_n)$ is a sufficient statistic with distribution
\[p\left(x\right) = \exp\left(2\theta' x\right) g\left(x\right),\]
where $g\left(x\right)$ is a function that counts the number of possible tournament results giving the net win vector $x$. A bijection proof shows that $x$ is indeed Schur-concave. Therefore, we can use Procedures 1--3 to compare player qualities.

After conditioning on $U(X) = (X_1 + X_2, X_3, \ldots, X_n)$, and under the assumption $\theta_1 = \theta_2$, every feasible configuration of $Y$ is equally likely. If $n$ is not too large (say, no more than $40$ players), we can find the conditional distribution of $X_1$ by enumerating over the configurations; for larger $n$, computation might pose a more serious problem, requiring us for example to compute the $p$-value using Markov Chain Monte Carlo techniques \citep{Besag:1989}.
\end{example}

\begin{example}[Comparing the variances of different normal populations] Suppose there are $n$ normal populations with laws $N(\mu_j, \sigma_j^2)$ and $m$ independent observations from each of them. The sample variance for population $j$ can be denoted as $R_j$. By Cochran's theorem, $\left(m-1\right) R_j \sim \sigma_j^2 \chi_{m-1}^2$, and thus the joint distribution of $R$ is
\begin{align*}
r & \sim \prod_{j=1}^n \left(\frac{\left(m-1\right) r_j}{\sigma_j^2}\right)^{\left(m-3\right) / 2} e^{-\left(m-1\right) r_j / 2 \sigma_j^2} 1_{\left\{r_j > 0 \right\}} \\
& \propto \exp\left(-\frac{m-1}{2 \sigma_1^2} r_1 - \cdots - \frac{m-1}{2 \sigma_n^2} r_n\right) \prod_{j=1}^n r_j^{\left(m-3\right) / 2} 1_{\left\{r > 0\right\}}.
\end{align*}
The carrier measure is $\prod_{j=1}^n r_j^{\left(m-3\right) / 2} 1_{\left\{r > 0\right\}}$, which is Schur-concave. Thus, we can use Procedures 1--3 to find populations with the smallest or largest variances. In this example, the distribution of $X_1 / (X_1 + X_2)$ conditional on $(X_1 + X_2, X_3, \ldots, X_4)$ is distributed as $\text{Beta}(m/2, m/2)$ under the null, or equivalently $X_1 / X_2$ is conditionally distributed as $F_{m, m}$; hence a (two-tailed) $F$-test is valid for comparing the top two populations.
\end{example}

\section{Verifying the Winner: Is the Winner Really the Best?}
\label{sec:winnerbest}

First, we justify the notion that the population with largest $\theta_j$ is also the largest population in stochastic order:
\begin{theorem}
\label{thm:stoch}
For a multivariate exponential family with a symmetric carrier distribution, $X_1 \ge X_2$ in stochastic order if and only if $\theta_1 \geq \theta_2$.
\end{theorem}

\begin{proof}
It suffices to prove the ``if'' part, as the ``only if'' part can be follows from swapping the role of $\theta_1$ and $\theta_2$. For any fixed $a$, and $x_1 \ge a$ and $x_2 < a$, we have $x_1 > x_2$ and
\[\exp\left(\theta_1 x_1 + \theta_2 x_2 + \cdots + \theta_n x_n - \psi\left(\theta\right)\right) g\left(x\right) \ge \exp\left(\theta_1 x_2 + \theta_2 x_1 + \cdots + \theta_n x_n - \psi\left(\theta\right)\right) g\left(x\right).\]
Integrating both sides over the region $\left\{x:\; x_1 \ge a, x_2 < a\right\}$ gives
\[\PP\left[X_1 \ge a, X_2 < a\right] \ge \PP\left[X_1 < a, X_2 \ge a\right].$$
Now adding $\PP\left[X_1 \ge a, X_2 \ge a\right]$ to both probabilities gives
$$\PP\left[X_1 \ge a\right] \ge \PP\left[X_2 \ge a\right],\]
meaning that $X_1$ is greater than $X_2$ in stochastic order.
\end{proof}

Before proving our main result for Procedure 1, we give the following lemmas, the first of which clarifies a key idea in the proof, and the second is needed for a sharper bound in \eqref{eq:marginal}.

\begin{lemma}[\citealp{Berger:1982hy}]
\label{lma:union}
If $p_j$ are valid $p$-values for testing null hypothesis $H_{0j}$, then $p_* = \max_j p_j$ is a valid $p$-value for the union null (i.e.\ disjunction null) hypothesis $H_0 = \bigcup_j H_{0j}$.
\end{lemma}

\begin{proof}
Under $H_0$, one of the $H_{0j}$ is true; without loss of generality assume it is $H_{01}$. Then,
\[\PP\left[p_* \le \alpha\right] \le \PP\left[p_1 \le \alpha\right] \le \alpha.\]
Therefore $p_*$ is a valid $p$-value for the union null hypothesis.
\end{proof}

\begin{lemma}
\label{lma:margsharp}
If $\theta_1 \ge \max_{j \ne 1} \theta_j$, then $\PP\left[1 \text{ wins}\right] \ge \frac{1}{n}$.
\end{lemma}

\begin{proof}
We can prove so with a coupling argument: for any sequence $x_1, x_2, \ldots, x_n$, define $\tau\left(x\right) = \left\{\tau\left(x_j\right)\right\}_{j = 1, \ldots, n}$, obtained by swapping $x_1$ with the largest value in the sequence $x$. Hence
\[\exp\left(\theta_1 \tau\left(x_1\right) + \cdots + \theta_n \tau\left(x_n\right) - \psi\left(\theta\right)\right) g\left(X\right) \ge \exp\left(\theta_1 x_1 + \cdots + \theta_n x_n - \psi\left(\theta\right)\right) g\left(X\right).\]
If we integrating both sides over $\mathbb{R}^n$ (or $\mathbb{Z}^n$ in the case of counting measure), the right hand side gives $1$. Since $\tau$ is an $n$-to-$1$ mapping, the left hand side is $n$ times the integral over $\left\{x_1 \ge \max_{j > 1} x_j\right\}$. In other words,
\[n \PP\left[1 \text{ wins}\right] \ge 1\]
as desired.

In the case of counting measure, the above argument follows if a subscript is attached to identical observations uniformly to ensure strict ordering.
\end{proof}

We are now ready to prove our result for Procedure 1, restated here for reference.

\begin{customthm}[Part 1 of \Cref{thm:main}]
Assume the model~\eqref{eq:expfam} holds and $g\left(x\right)$ is a Schur-concave function. Procedure 1 (the unadjusted pairwise test) has level $\alpha$ conditional on the best population not winning.
\end{customthm}

\begin{proof}
Let $j^*$ denote the (fixed) index of the best population, so $\theta_{j^*} \ge \max_{j\neq j^*} \theta_j$. The type I error --- the probability of incorrectly declaring any other $j$ to be the best --- is
\[\PP\left[\bigcup_{j \ne j^*} \text{declare } j \text{ best}\right] \le \sum_{j \ne j^*} \PP\left[\text{declare } j \text{ best} \;\middle|\; j \text{ wins}\right] \PP\left[j \text{ wins}\right],\]
recalling that ties are broken randomly, so there is only one winner in any realization. Thus, it is enough to bound $\PP_{\theta}\left[\text{declare } j \text{ best} \;\middle|\; j \text{ wins}\right] \le \alpha$, for each $j \ne j^*$, and for all $\theta$ with $j^* \in \argmax_j \theta_j$. Then we will have
\begin{equation}
\label{eq:marginal}
\PP\left[\bigcup_{j \ne j^*} \text{declare } j \text{ best}\right] \le \sum_{j \ne j^*} \alpha \cdot \PP\left[j \text{ wins}\right] = \alpha \PP\left[j^* \text{ does not win}\right] \le \frac{n-1}{n} \alpha,
\end{equation}
where the last inequality follows from \Cref{lma:margsharp}.

We start by assuming that we are working with the Lebesgue measure rather than the counting measure (eliminating the possibility of ties). The necessary modification of the proof for the counting measure case is provided at the end of this proof.

To minimize notational clutter, we consider only the case where the winner is $1$, i.e.\ $X_1 \ge \max_{j > 1} X_j$. Furthermore, we will denote the runner-up with $2$. This is not necessarily true, but we will use it as a shorthand to simplify our notation. For other cases, the following proof remains valid under relabeling and can thus be applied. In this case, we will test the null hypothesis $H_{01}:\; \theta_1 \le \max_{j > 1} \theta_j$, which is the union of the null hypotheses $H_{01j}:\; \theta_1 \le \theta_j$ for $j \ge 2$. For each of these we can construct an exact $p$-value $p_{1j}$, which is valid under $H_{01j}$ conditional on $A_1$, the event that $X_1$ is the winner. Hence by \Cref{lma:union}, a test that rejects when $p_{1*} = \max_j p_{1j} \le \alpha$ is valid for $H_{01}$ conditional on $A_1$. Procedure 1 performs an unadjusted pairwise test comparing $X_1$ to $X_2$. Hence it is sufficient to show that $p_{12} = p_{1*}$ and that rejecting when $p_{12}\le \alpha$ coincides with the unadjusted pairwise test.

Our proof has three main parts:
\begin{inlinelist}
\item deriving $p_{1j}$ for each $j\ge 2$,
\item showing that $p_{12} \ge p_{1j}$ for each $j\ge 2$, and
\item showing that $p_{12}$ is an unadjusted pairwise $p$-value.
\end{inlinelist}

\paragraph{Derivation of $p_{1j}$}

Following the framework in \citet{Fithian:2014ws}, we first construct the $p$-values by conditioning on the selection event where the winner is $1$: 
\[A_1 = \left\{X_1 \ge \max_{j > 1} X_j\right\}.\]
For convenience, we let
\[D_{jk} = \frac{X_j - X_k}{2} \quad \text{ and } \quad M_{jk} = \frac{X_j + X_k}{2}.\]

We then re-parametrize to replace $X_1$ and $X_j$ with $D_{1j}$ and $M_{1j}$. The distribution is now an exponential family with sufficient statistics $D_{1j}, M_{1j}, X_{\setminus\left\{1, j\right\}}$ and corresponding natural parameters $\theta_1 - \theta_j, \theta_1 + \theta_j, \theta_{\setminus\left\{1, j\right\}}$. We now consider
\begin{equation}
\label{eq:condlaw}
\mathcal{L}_{\theta_1 - \theta_j = 0} \left(D_{1j} \;\middle|\; M_{1j}, X_{\setminus\left\{1, j\right\}}, A_1\right).
\end{equation}
We can rewrite the selection event in terms of our new parameterization as
\begin{align*}
A_1 &= \left\{X_1 \ge X_j\right\} \cap \left\{X_1 \ge \max_{k \neq 1,j} X_k\right\}\\
&= \left\{D_{1j} \ge 0\right\} \cap \left\{D_{1j} \ge \max_{k \ne 1, j} X_k - M_{1j}\right\}.
\end{align*}
The conditional law of $D_{1j}$ in \eqref{eq:condlaw}, in particular, is a truncated distribution.
\begin{align*}
p\left(d_{1j} \mid M_{1j}, X_{\setminus\left\{1, j\right\}}, A_1\right) & \propto \exp\left(\left(\theta_1 - \theta_j\right) d_{1j} + \theta_2 X_2 + \cdots + \left(\theta_1 + \theta_j\right) M_{1j} + \cdots + \theta_n X_n \right) \\
& \quad\quad g\left(M_{1j} + d_{ij}, X_2, \ldots, M_{ij} - d_{ij}, \ldots X_n\right) 1_{A_1} \\
& \stackrel{\text{(a)}}{\propto} g\left(M_{1j} + d_{1j}, X_2, \ldots, M_{1j} - d_{1j}, \ldots X_n\right) 1_{A_1},
\end{align*}
where at step (a), conditioning on $X_{\setminus\left\{1, j\right\}}$ and $M_{1j}$ removes dependence on $\theta_{\setminus\left\{1, j\right\}}$ and $\theta_1 + \theta_j$ respectively, while $\theta_1 - \theta_j$ is taken to be $0$ under our null hypothesis. Note that we consider this as a one-dimensional distribution of $D_{1j}$ on $\RR$, where $M_{1j}$ and $X_{\setminus\left\{1, j\right\}}$ are treated as fixed.

The $p$-value for $H_{01j}$ is thus
\begin{equation}
\label{eq:p1j}
p_{1j} = \frac{\int_{D_{1j}}^\infty g\left(M_{1j} + z, X_2, \ldots, M_{1j} - z, \ldots, X_n\right) \,dz}{\int_{\max\left\{X_2 - M_{1j}, 0\right\}}^\infty g\left(M_{1j} + z, X_2, \ldots, M_{1j} - z, \ldots, X_n\right) \,dz}.
\end{equation}

Finally, by construction, $p_{1j}$ satisfies
\[\PP_{H_{01j}}\left[p_{1j} < \alpha \;\middle|\; M_{1j}, X_{\setminus\left\{1, j\right\}}, A_1\right] \le \alpha \quad \text{a.s.},\]
Marginalizing over $M_{1j}, X_{\setminus\left\{1, j\right\}}$,
\[\PP_{H_{01j}}\left[p_{1j} < \alpha \;\middle|\; A_1\right] \le \alpha.\]
Therefore these $p_{1j}$ are indeed valid $p$-values.

\paragraph{Demonstration that $p_{1*}=p_{12}$}

We now proceed to show that $p_{12}$, the $p$-value comparing the winner to the runner-up, is the largest of all $p_{1j}$. Without loss of generality, it is sufficient to show that $p_{12} \ge p_{13}$.

From the first part of this proof, both $p$-values are constructed by conditioning on $X_{\setminus\left\{1, 2, 3\right\}}$. Upon conditioning these, $\left(X_1, X_2, X_3\right)$ follows an exponential family distribution, with carrier distribution
\[g_{X_4, \ldots, X_n}\left(X_1, X_2, X_3\right) = g\left(X_1, \ldots, X_n\right),\]
here $X_4, \ldots, X_n$ are used in the subscript as they are conditioned on and no longer considered as variables. The first point in \Cref{lma:twoprop} says that the function $g_{X_4, \ldots, X_n}$ is Schur-concave as well. We have reduced the problem to the case when $n = 3$: we can apply the result for $n = 3$ to $g_{X_4, \ldots, X_n}$ to yield $p_{12} \ge p_{13}$ for $n > 3$.

We have reduced to the case when $n = 3$. The p-values thus are
\begin{align*}
p_{12} & = \frac{\int_{D_{12}}^\infty g\left(M_{12} + z, M_{12} - z, X_3\right) \,dz}{\int_0^\infty g\left(M_{12} + z, M_{12} - z, X_3\right) \,dz}, \\
p_{13} & = \frac{\int_{D_{13}}^\infty g\left(M_{13} + z, X_2, M_{13} - z\right) \,dz}{\int_{\max\left\{X_2 - M_{13}, 0\right\}}^\infty g\left(M_{13} + z, X_2, M_{13} - z\right) \,dz}
\end{align*}

The maximum in the denominator of $p_{13}$ prompts us to consider two separate cases. First, we suppose $X_2 < M_{13}$. Changing variables such that the lower limits of both integrals in the numerator are $0$, we can re-parametrize the integrals above to give
\begin{align*}
p_{12} & = \frac{\int_0^\infty g\left(X_1 + z, X_2 - z, X_3\right) \,dz}{\int_0^\infty g\left(M_{12} + z, M_{12} - z, X_3\right) \,dz} \\
& = \frac{\int_0^\infty g\left(X_1 + z, X_2 - z, X_3\right) \,dz}{\int_{-D_{12}}^\infty g\left(X_1 + z, X_2 - z, X_3\right) \,dz}, \\
p_{13} & = \frac{\int_0^\infty g\left(X_1 + z, X_2, X_3 - z\right) \,dz}{\int_0^\infty g\left(M_{13} + z, X_2, M_{13} - z\right) \,dz} \\
& = \frac{\int_0^\infty g\left(X_1 + z, X_2, X_3 - z\right) \,dz}{\int_{-D_{13}}^\infty g\left(X_1 + z, X_2, X_3 - z\right) \,dz}.
\end{align*}
To help see the re-parametrization, each of these integrals can be thought of in terms of integrals along segments and rays. For example $p_{12}$ can be represented in terms of integrals $A$ and $B$ in \Cref{fig:p-value}. Specifically,
\[p_{12} = \frac{B}{A+B}\]

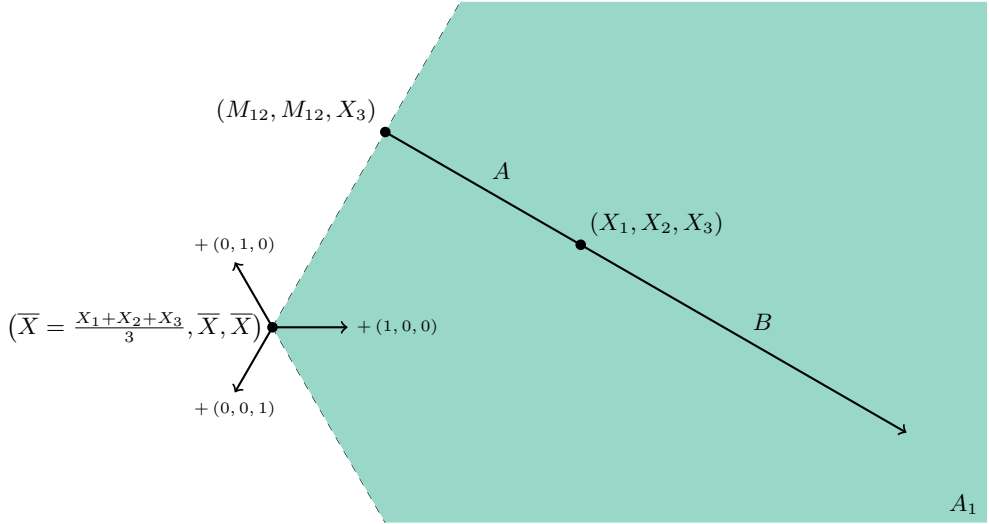
\begin{figure}[htbp]
\centering
\begin{tikzpicture}
	\coordinate (C1) at (60:5);
	\coordinate (C2) at (300:3);
	\coordinate (A) at (60:3);
	\coordinate (B) at ($(A) + (330:3)$);
	\draw[dashed] (0, 0) -- (C1);
	\draw[dashed] (0, 0) -- (C2);
	\fill[mid_green] (0, 0) -- (C1) -- ($(C1) + (7, 0)$) -- ($(C2) + (8, 0)$) -- (C2) -- cycle;
	\node[above left] at ($(C2) + (8, 0)$) {$A_1$};
	\draw[thick] (A) -- (B) node[midway, above right] {$A$};
	\draw[thick, ->] (B) -- ++(330:5) node[midway, above right] {$B$};
	\fill (A) circle(2pt) node[above left] {$\left(M_{12}, M_{12}, X_3\right)$};
	\fill (B) circle(2pt) node[above right] {$\left(X_1, X_2, X_3\right)$};
	\fill (0, 0) circle(2pt) node[left] {$\left(\overline{X} = \frac{X_1 + X_2 + X_3}{3}, \overline{X}, \overline{X}\right)$};
	\draw[thick, ->] (0, 0) -- (1, 0) node[right] {\tiny $+ \left(1, 0, 0\right)$};
	\draw[thick, ->] (0, 0) -- (120:1) node[above] {\tiny $+ \left(0, 1, 0\right)$};
	\draw[thick, ->] (0, 0) -- (240:1) node[below] {\tiny $+ \left(0, 0, 1\right)$};
\end{tikzpicture}
\caption{The $p$-value $p_{12}$ can be written in terms of integral $A$ along the segment and $B$ along the ray. The diagram is drawn a level set of $x_1 + x_2 + x_3$. The green region represents the selection event $A_1$.}
\label{fig:p-value}
\end{figure}

\Cref{fig:comparerays} has both the $p$-values shown on the same diagram. Proving $p_{12} \ge p_{13}$ is the same as proving
\[\frac{B}{A+B} \ge \frac{D}{C+D} \quad \Longleftrightarrow \quad \frac{B}{A} \ge \frac{D}{C}.\]
We will prove so by extending $A$ to include $\tilde{A}$ on the diagram. We denote the sum $A + \tilde{A}$ as $A'$. Formally,
\begin{equation}
\label{eq:int_extension}
A' = \int_{-D_{13}}^0 g\left(X_1 + z, X_2 - z, X_3\right) \,dz \ge \int_{-D_{12}}^0 g\left(X_1 + z, X_2 - z, X_3\right) \,dz = A.
\end{equation}
It is thus sufficient to show that $B \ge D$ and $C \ge A'$.

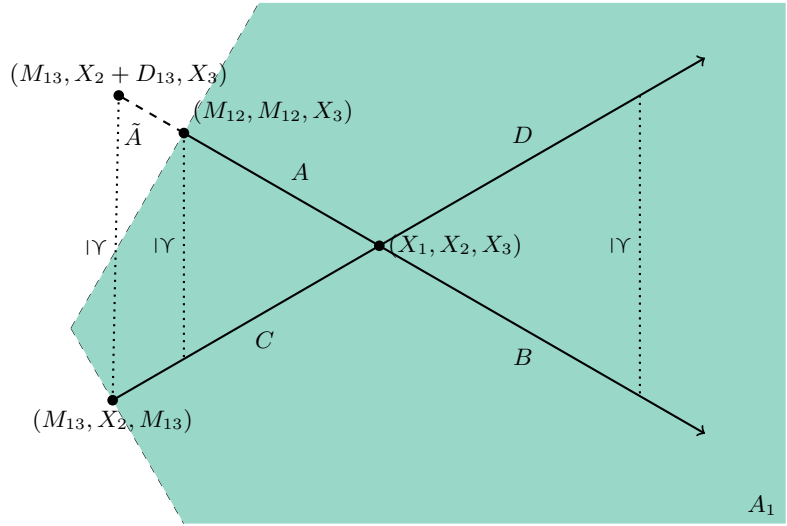
\begin{figure}[htbp]
\centering
\begin{tikzpicture}
	\coordinate (C1) at (60:5);
	\coordinate (C2) at (300:3);
	\coordinate (A) at (60:3);
	\coordinate (B) at ($(A) + (330:3)$);
	\coordinate (C) at ($(0, 0)!(B)!(C2)$);
	\coordinate (D) at ($(A) + (150:1)$);
	\draw[dashed] (0, 0) -- (C1);
	\draw[dashed] (0, 0) -- (C2);
	\fill[mid_green] (0, 0) -- (C1) -- ($(C1) + (7, 0)$) -- ($(C2) + (8, 0)$) -- (C2) -- cycle;
	\node[above left] at ($(C2) + (8, 0)$) {$A_1$};
	\draw[thick] (A) -- (B) node[midway, above right] {$A$};
	\draw[thick, ->] (B) -- ++(330:5) node[midway, below left] {$B$};
	\fill (A) circle(2pt) node[above right] {$\left(M_{12}, M_{12}, X_3\right)$};
	\fill (B) circle(2pt) node[right] {$\left(X_1, X_2, X_3\right)$};
	\draw[thick] (C) -- (B) node[midway, below right] {$C$};
	\node[below] at (C) {$\left(M_{13}, X_2, M_{13}\right)$};
	\draw[thick, ->] (B) -- ++(30:5) node[midway, above left] {$D$};
	\fill ($(0, 0)!(B)!(C2)$) circle(2pt);
	\draw[thick, dashed] (A) -- (D) node[midway, below left] {$\tilde{A}$};
	\fill (D) circle(2pt) node[above] {$\left(M_{13}, X_2 + D_{13}, X_3\right)$};
	\draw[thick, dotted] (C) -- (D) node[midway, below, rotate=-90] {$\succeq$};
	\draw[thick, dotted] (A) -- ++(0, -3) node[midway, below, rotate=-90] {$\succeq$};
	\draw[thick, dotted] ($(B) + (30:4)$) -- ($(B) + (330:4)$) node[midway, below, rotate=-90] {$\succeq$};
\end{tikzpicture}
\caption{The $p$-value $p_{12}$ can be written in terms of integral $A$ along the segment and $B$ along the ray; and $p_{13}$ in terms of $C$ and $D$. $A'$ would refer to the sum of $A$ with the dashed line portion labeled as $\tilde{A}$, formally explained in \Cref{eq:int_extension}. The majorization relation is indicated by the dotted line.}
\label{fig:comparerays}
\end{figure}

Indeed from the second point in \Cref{lma:twoprop} we have
\[\left(X_1 + z, X_2 - z, X_3\right) \succeq \left(X_1 + z, X_2, X_3 - z\right)\]
for $z \le 0$ and the majorization reversed for $z \ge 0$. This majorization relation is indicated as the dotted line in \Cref{fig:comparerays}. So Schur-concavity shows that
\[g\left(X_1 + z, X_2 - z, X_3\right) \le g\left(X_1 + z, X_2, X_3 - z\right)\]
for $z \le 0$, and the inequality reversed for $z \ge 0$. Taking integrals on both sides yields the desired inequality.

For the second case where $X_2 \ge M_{13}$, the segment $C$ will reach the line $x_1 = x_2$ first before it reaches $x_1 = x_3$, ending at $\left(X_2, X_2, X_1 - X_2 + X_3\right)$ instead. But we can still extend $A$ by $\tilde{A}$ to $\left(X_2, X_1, X_3\right)$. The rest of the proof follows. In either cases, $p_{12} \ge p_{13}$, or in generality, $p_{12} \ge p_{1j}$ for $j > 1$. In other words, $p_{12} = p_{1*}$.

\paragraph{$p_{12}$ is an unadjusted pairwise $p$-value}

Before conditioning on $A_1$, the distribution in \eqref{eq:condlaw} is symmetric around $0$ under $\theta_1=\theta_j$. Since the denominator of $p_{12}$ integrates over half of this symmetric distribution, it is always equal to $1/2$. Thus, the one-sided conditional test at level $\alpha$ is equivalent to the one-sided unadjusted test at level $\alpha/2$, or equivalently the two-sided unadjusted pairwise test at level $\alpha$.

\paragraph{Modification for counting measure}

Now suppose the exponential family is defined on the counting measure instead. If ties are broken independently and randomly, the end points on the rays can be considered as ``half an atom'' if the coordinates are integers (or a smaller fraction of an atom in case of a multi-way tie). The number of atoms on each ray is the same (after the extension $\tilde{A}$) and the atoms on each ray can be paired up in exactly the same way as illustrated in \Cref{fig:comparerays}, with the inequalities above still holding for each pair of the atoms. Summing these inequalities yields our desired result.
\end{proof}

\subsection{Power Comparison in the Multinomial Case}
\label{sec:subsetsel}

As the construction of this test follows \citet{Fithian:2014ws}, it uses UMPU selective level-$\alpha$ tests for the pairwise $p$-values. This section compares the power of our procedure to the best previously known method for verifying multinomial ranks, by \citet{Gupta:1967wg}. They devise a rule to select a subset that includes the maximum $\pi_j$. In other words, if the selected subset is $J\left(X\right)$, it guarantees
\begin{equation}
\label{eq:subsetsel}
\PP\left[\argmax_j \pi_j \in J\left(X\right)\right] \ge 1 - \alpha.
\end{equation}
This is achieved by finding an integer $d$, as a function on $m$, $n$ and $\alpha$, and selecting the subset
\[J\left(X\right) = \left\{j: X_j \ge \max_k X_k - d\right\}.\]
We take $d(m, n, \alpha)$ to be the smallest integer such that \eqref{eq:subsetsel} holds for any $\pi$; \citet{Gupta:1967wg} provide an algorithm for determining $d$.

Subset selection is closely related to testing whether the winner is the best. In particular, we can define a test that declares $j$ the best whenever $J\left(X\right) = \left\{j\right\}$. If $J\left(X\right)$ satisfies \eqref{eq:subsetsel}, this test is valid at level $\alpha$. We next compare the power of the resulting test against the power of our Procedure 1 in a multinomial example with $\pi \propto \left(e^\delta, 1, \ldots, 1\right)$, for several combinations of $m$ and $n$.

\Cref{fig:power} gives the power curves for $\text{Multinomial}\left(m, \pi\right)$ and
\[\pi \propto \left(e^\delta, 1, \ldots, 1\right),\]
for various combinations of $m$ and $n$. For their method, we use $\alpha = 0.05$; but in light of the extra factor of $\frac{n-1}{n}$ in \eqref{eq:marginal}, we will apply the selective procedure with $\frac{n}{n-1} \alpha$ such that the marginal type I error rate of both procedures are controlled at $\alpha$. Their test coincides with our test at $n = 2$; however as $n$ grows, the selective test shows significantly more power than \citeauthor{Gupta:1967wg}'s test.

\begin{figure}[htbp]
\centering
\includegraphics[width = \textwidth]{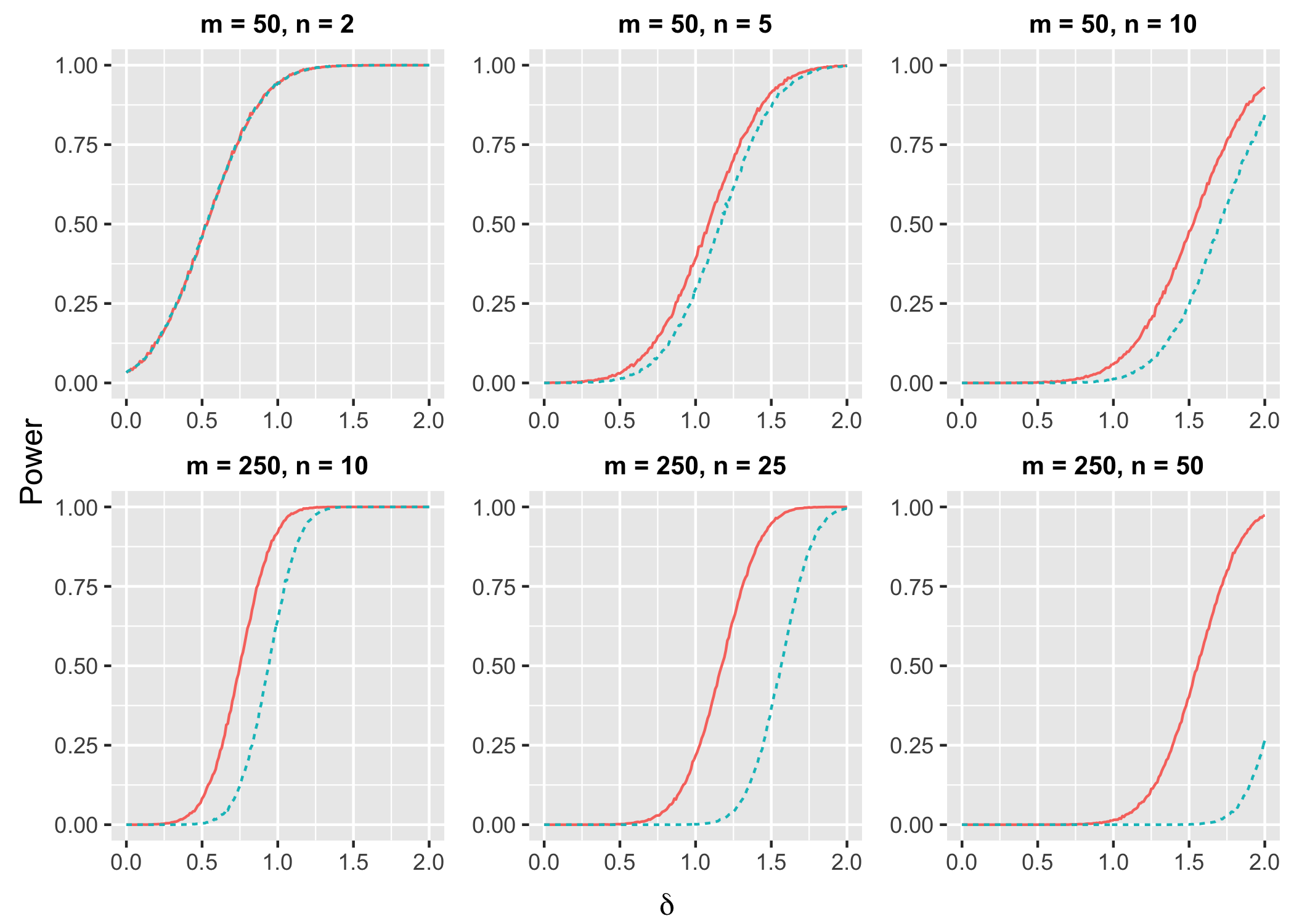}
\caption{Power curves as a function of $\delta$. The plots in the first row all have $m = 50$ and the second row $m = 250$. The solid line and the dashed line are the power for the selective test and Gupta and Nagel's test, respectively.}
\label{fig:power}
\end{figure}

To interpret, e.g., the upper right panel of \Cref{fig:power}, suppose that in a poll of $m = 50$ respondents, one candidate enjoys $30\%$ support and the other $n - 1 = 9$ split the remainder ($\delta = \log\frac{0.3}{0.7 / 9} \approx 1.35$). Then our procedure has power approximately $0.3$ to detect the best candidate, while \citeauthor{Gupta:1967wg}'s procedure has power around $0.1$.

To understand why our method is more powerful, note that both procedures operate by comparing $X_{[1]}-X_{[2]}$ to some threshold, but the two methods differ in how that threshold is determined. The threshold from \citet{Gupta:1967wg} is fixed and depends on $m$ and $n$ alone, whereas in our procedure the threshold depends on $X_{[1]}+X_{[2]}$, a data-adaptive choice. 

The difference between the two methods is amplified when $n$ is large and $\pi_{(1)} \ll 1/2$. In that case, $d$ from \citeauthor{Gupta:1967wg} is usually computed based on the worst-case scenario $\pi = \left(\frac{1}{2}, \frac{1}{2}, 0, \ldots, 0\right)$; i.e.\ $d$ is the upper $\alpha$ quantile of
\[X_1 - X_2 \sim m - 2 \cdot \text{Binomial}\left(m, \frac{1}{2}\right) \approx \text{Normal}\left(0, m\right).\]
Thus $d \approx \sqrt{m} z_\alpha$, where $z_\alpha$ is the upper $\alpha$ quantile of a standard Gaussian. On the other hand, our method defines a threshold based on the upper $\frac{n}{n-1} \cdot \frac{\alpha}{2}$ quantile of
\[X_1 - X_2 \mid X_1 + X_2 \sim X_1 + X_2 - 2 \cdot \text{Binomial}\left(X_1 + X_2, \frac{1}{2}\right),\]
which is approximately $\sqrt{X_1 + X_2} z_{\alpha / 2}$. If $\pi_{(1)} \ll 1/2$ then with high probability $X_1 + X_2 \ll  m$, making our test much more liberal.

\section{Confidence Bounds on Differences: By How Much?}
\label{sec:confbound}

By generalizing the above, we can construct a lower confidence bound for $\theta_{[1]} - \max_{j \ne [1]} \theta_{j}$. Here we provide a more powerful Procedure 2' first. We will proceed by inverting a statistical test of the hypothesis $H_{0[1]}^\delta:\; \theta_{[1]} - \max_{j \ne [1]} \theta_{j} \le \delta$, which can be written as a union of null hypotheses:
\[H_{0[1]}^\delta = \bigcup_{j \ne [1]} H_{0[1]j}: \theta_{[1]} - \theta_j \le \delta.\]
By \Cref{lma:union}, we can construct selective exact one-tailed $p$-values $p_{[1]j}^\delta$ for each of these by conditioning on $A_{[1]}$, $M_{[1]j}$ and $X_{\setminus\left\{[1], j\right\}}$, giving us an exact test for $H_{0[1]}$ by rejecting whenever $\max_{j \ne [1]} p_{[1]j}^\delta < \alpha$.

\begin{theorem}
The $p$-values constructed above satisfy $p_{[1][2]}^\delta \ge p_{[1]j}^\delta$ for any $j \ne [1]$.
\end{theorem}

\begin{proof}
Again we start with assuming $X_1 \ge X_2 \ge \max_{j > 2} X_j$ for convenience. The $p$-values in question are derived from the conditional law
\[\mathcal{L}_{\theta_1 - \theta_j = \delta} \left(D_{1j} \;\middle|\; M_{1j}, X_2, \ldots, X_n, A\right),\]
which is the truncated distribution
\begin{align*}
p\left(d_{1j}\right) & \propto \exp\left(\left(\theta_1 - \theta_j\right) d_{1j} + \theta_2 X_2 + \cdots + \left(\theta_1 + \theta_j\right) M_{1j} + \cdots + \theta_n X_n \right) \\
& \quad\quad g\left(M_{1j} + d_{1j}, X_2, \ldots, M_{1j} - d_{1j}, \ldots X_n\right) 1_{A_1} \\
& \propto \exp\left(\delta d_{1j}\right) g\left(M_{1j} + d_{1j}, X_2, \ldots, M_{1j} - d_{1j}, \ldots X_n\right) 1_{A_1}.
\end{align*}

The $p$-values thus are
\[p_{1j}^\delta = \frac{\int_{D_{1j}}^\infty \exp\left(\delta z\right) g\left(M_{1j} + z, X_2, \ldots, M_{1j} - z, \ldots, X_n\right) \,dz}{\int_{\max\left\{X_2 - M_{1j}, 0\right\}}^\infty \exp\left(\delta z\right) g\left(M_{1j} + z, X_2, \ldots, M_{1j} - z, \ldots, X_n\right) \,dz}.\]

As before in Part 1 of \Cref{thm:main}, the conditioning reduces to the case where $n = 3$. Once again it is sufficient to show that $p_{12} \ge p_{13}$. We have the same two cases. If $X_2 < M_{13}$, then
\begin{align*}
p_{12}^\delta & = \frac{\int_0^\infty \exp\left(\delta \left(z + D_{12}\right)\right) g\left(X_1 + z, X_2 - z, X_3\right) \,dz}{\int_{-D_{12}}^\infty \exp\left(\delta \left(z + D_{12}\right)\right) g\left(X_1 + z, X_2 - z, X_3\right) \,dz} \\
& = \frac{\int_0^\infty \exp\left(\delta z\right) g\left(X_1 + z, X_2 - z, X_3\right) \,dz}{\int_{-D_{12}}^\infty \exp\left(\delta z\right) g\left(X_1 + z, X_2 - z, X_3\right) \,dz} \\
p_{13}^\delta & = \frac{\int_0^\infty \exp\left(\delta \left(z + D_{13}\right)\right) g\left(X_1 + z, X_2, X_3 - z\right) \,dz}{\int_{-D_{13}}^\infty \exp\left(\delta \left(z + D_{13}\right)\right) g\left(X_1 + z, X_2, X_3 - z\right) \,dz} \\
& = \frac{\int_0^\infty \exp\left(\delta z\right) g\left(X_1 + z, X_2, X_3 - z\right) \,dz}{\int_{-D_{13}}^\infty \exp\left(\delta z\right) g\left(X_1 + z, X_2, X_3 - z\right) \,dz}.
\end{align*}

The same argument in \Cref{fig:comparerays} shows that $p_{12}^\delta \ge p_{13}^\delta$. This is again true for the case where $X_2 \ge M_{13}$ as well.
\end{proof}

In other words, Procedure 2' can be summarized as: Find the minimum $\delta$ such that $p_{[1][2]}^\delta \le \alpha$. And by construction, Procedure 2' gives exact $1-\alpha$ confidence bound for $\theta_{[1]} - \max_{j \ne [1]} \theta_j$.

\begin{customthm}[Part 2 of \Cref{thm:main}]
Assume the model~\eqref{eq:expfam} holds and $g\left(x\right)$ is a Schur-concave function. Procedure 2 (the lower bound of unadjusted pairwise confidence interval) gives a conservative $1-\alpha$ lower confidence bound for $\theta_{[1]} - \max_{j \ne [1]} \theta_j$.
\end{customthm}

\begin{proof}
When Procedure 2 reports $-\infty$ as a confidence lower bound, it is definitely valid and conservative. It remains to show that when Procedure 2 reports a finite confidence lower bound, it is smaller than the confidence lower bound reported by Procedure 2'.

If Procedure 2 reports a finite confidence lower bound $\delta^*$, then $\delta^* \ge 0$. Also
\begin{equation}
\label{eq:naivelo}
\frac{\alpha}{2} = \frac{\int_{D_{12}}^\infty \exp\left(\delta^* z\right) g\left(M_{12} + z, X_2, \ldots, M_{12} - z, \ldots, X_n\right) \,dz}{\int_{-\infty}^\infty \exp\left(\delta^* z\right) g\left(M_{12} + z, X_2, \ldots, M_{12} - z, \ldots, X_n\right) \,dz}
\end{equation}
as Procedure 2 is constructed from an unadjusted two-tail pairwise confidence interval. However, as $\delta^* \ge 0$, we have
\begin{align*}
\frac{\int_{-\infty}^0 \exp\left(\delta^* z\right) g\left(M_{12} + z, X_2, \ldots, M_{12} - z, \ldots, X_n\right) \,dz}{\int_0^\infty \exp\left(\delta^* z\right) g\left(M_{12} + z, X_2, \ldots, M_{12} - z, \ldots, X_n\right) \,dz} & \le 1 \\
\frac{\int_{-\infty}^\infty \exp\left(\delta^* z\right) g\left(M_{12} + z, X_2, \ldots, M_{12} - z, \ldots, X_n\right) \,dz}{\int_0^\infty \exp\left(\delta^* z\right) g\left(M_{12} + z, X_2, \ldots, M_{12} - z, \ldots, X_n\right) \,dz} & \le 2.
\end{align*}
Multiplying this to \eqref{eq:naivelo}, we have
\[\alpha \ge \frac{\int_{D_{12}}^\infty \exp\left(\delta^* z\right) g\left(M_{12} + z, X_2, \ldots, M_{12} - z, \ldots, X_n\right) \,dz}{\int_0^\infty \exp\left(\delta^* z\right) g\left(M_{12} + z, X_2, \ldots, M_{12} - z, \ldots, X_n\right) \,dz},\]
indicating that $\delta^*$ is smaller than the confidence bound that Procedure 2' would report. Hence $\delta^*$ is a valid and conservative.
\end{proof}

Note that Procedure 2 reporting $-\infty$ in case of $\delta^* \le 0$ is rather extreme. In reality, we can always just adopt Procedure 2' in the case when Procedure 1 rejects. In fact, by Procedure 2', the multinomial example for polling in \Cref{sec:iowa} can give a stronger lower confidence bound, that $\pi_{\text{Trump}} / \max_{j \ne \text{Trump}} \pi_{j} \ge 1.108$ (Trump leads the field by at least $10.8\%$).

\section{Verifying Other Ranks: Is the Runner-Up Really the Second Best, etc.?}
\label{sec:stepwise}

Often we will be interested in verifying ranks beyond the winner. More generally, we could imagine declaring that the first $j$ populations are all in the correct order, that is
\begin{equation}
\label{eq:jrankscorrect}
\theta_{[1]} > \cdots > \theta_{[j]} > \max_{k>j} \theta_{[k]}.
\end{equation}

Let $j_0$ denote the largest $j$ for which \eqref{eq:jrankscorrect} is true. Note that $j_0$ is both random and unknown, because it depends on both the data and population ranks. Procedure 3 declares that $j_0 \ge j$ if the unadjusted pairwise tests between $X_{[k]}$ and $X_{[k+1]}$, reject at level $\alpha$ for {\em all} of $k = 1, \ldots, j$.

In terms of the Iowa polling example of \Cref{sec:intro}, we would like to produce a statement of the form ``Trump has the most support, Cruz has the second-most, and Rubio has the third-most.'' Procedure 3 performs unadjusted pairwise tests to ask if Cruz is really the runner-up upon verifying that Trump is the best, and if Rubio is really the second runner-up upon verifying that Cruz is the runner-up, etc., until we can no longer infer that a certain population really holds its rank.

While we aim to declare more populations to be in the correct order, declaring too many populations, i.e.\ out-of-place populations, to be in the right order is undesirable. It is possible to consider false discovery rate (the expected portion of out-of-place populations declared) here, but we restrict our derivation to FWER (the probability of having any out-of-place populations declared).

Formally, let $\hat{j}_0$ denote the number of ranks validated by a procedure (the number of rejections). Then the FWER of $\hat{j}_0$ is the probability that too many rejections are made; i.e.\ $\PP\left[\hat{j}_0 > j_0\right]$. For example, suppose that the top three data ranks and population ranks coincide, but not the fourth ($j_0 = 3$). Then we will have made a Type I error if we declare that the top five ranks are correct ($\hat{j}_0 = 5$), but not if we declare that the top two are correct ($\hat{j}_0 = 2$). In other words, $\hat{j}_0$ is a lower confidence bound for $j_0$.

To show that Procedure 3 is valid, we will prove the validity of a more liberal Procedure 3', described in \Cref{alg:proc3}. Procedure 3 is equivalent to Procedure 3' for the most part, except that Procedure 3 conditions on a larger event $\left\{X_{[j]} \ge \max_{k>j} X_{[k]}\right\}$ in \Cref{lne:event}.

\begin{algorithm}[htbp]
	\SetKwInOut{Input}{input}
	\SetKwInOut{Output}{output}
	\SetKwData{Rejected}{rejected}
	\SetKwComment{Comment}{\# }{}
	\SetKw{Test}{test}
	\SetKw{True}{true}
	\Input{$X_1, \ldots, X_n$}
	\Output{$\hat{j}_0$, an estimate for $j_0$}
	\Comment{Initialization}
	$\tau_j \leftarrow [j]$\;
	\Comment{Consider $\tau_j$ as part of the observation and the fixed realization of the random index $[j]$}
	$X_{\tau_0} \leftarrow \infty$\;
	$j \leftarrow 0$\;
	\Rejected $\leftarrow$ \True\;
	\While{\Rejected}{
		$j \leftarrow j+1$\;
		$D_{\tau_j} \leftarrow X_{\tau_j} - X_{\tau_{j+1}}$\;
		\label{lne:event} Set up the distribution of $D_{\tau_j \tau_{j+1}}$, conditioned on
		\begin{itemize}
		\item the variables $X_{\tau_1}, \ldots, X_{\tau_{j-1}}, X_{\tau_{j+2}}, \ldots, X_{\tau_n]}$, and
		\item the event $\left\{X_{\tau_{j-1}} \ge X_{\tau_j} \ge \max_{k>j} X_{\tau_k)}\right\}$\;
		\end{itemize}
		\Comment{The distribution of $D_{\tau_j \tau_{j+1}}$ depends only on $\theta_{\tau_j} - \theta_{\tau_{j+1}}$ now}
		\Test $H_0: \theta_{\tau_j} - \theta_{\tau_{j+1}} \le 0$ against $H_1: \theta_{\tau_j} - \theta_{\tau_{j+1}} > 0$ according to the distribution of $D_{\tau_j \tau_{j+1}}$\;
		Set \Rejected as the output of the test\;
	}
	$\hat{j}_0 \leftarrow j-1$\;
	\caption{Procedure 3', a more liberal version of Procedure 3}
	\label{alg:proc3}
\end{algorithm}

\begin{theorem}
Procedure 3' is a stepwise procedure that an estimate $\hat{j}_0$ of $j_0$ at the FWER controlled at $\alpha$, where $j_0$ is given by
\[j_0 = \max_j \left\{\theta_{[1]} > \cdots > \theta_{[j]} > \max_{k>j} \theta_{[k]}\right\}.\]
\end{theorem}

\begin{proof}
We will first show that Procedure 3' falls into the sequential goodness-of-fit testing framework proposed by \citet{Fithian:2015uj}. We thus analyze Procedure 3' as a special case of the BasicStop procedure on random hypothesis, described in the same paper. This enables us to construct valid selective $p$-values and derive Procedure 3'.

\paragraph{Application of the sequential goodness-of-fit testing framework}

Upon observing $X_{[1]} \ge \cdots \ge X_{[n]}$, we can set up a sequence of nested models
\[\mathcal{M}_1\left(X\right) \subseteq \cdots \subseteq \mathcal{M}_n\left(X\right), \quad\text{ where } \mathcal{M}_j\left(X\right)^c = \left\{\theta:\; \theta_{[1]} > \cdots > \theta_{[j]} > \max_{k > j} \theta_{[k]}\right\}.\]
If we define the $j$-th null hypothesis as
\[\widetilde{H}_{0j}: \theta_{[j]} \le \max_{k>j} \theta_{[k]},\]
then $\widetilde{H}_{01},\ldots,\widetilde{H}_{0j}$ are all false if and only if $\theta\notin \mathcal{M}_j(X)$.

In other words, $\mathcal{M}_j\left(X\right)$ is a family of distributions that does not have all first $j$ ranks correct. As we will see later, each step in Procedure 3' is similar to testing $\widetilde{H}_{0j}$, stating that without the first $j$ ranks correct, it is hard to explain the observations. Thus, returning $\hat{j}_0 = j$ amounts to rejecting $\widetilde{H}_{01}, \ldots, \widetilde{H}_{0j}$, or equivalently determining that the models $\mathcal{M}_1(X), \ldots, \mathcal{M}_j(X)$ do {\em not} fit the data.

While the null hypotheses $\widetilde{H}_{0j}$ provided intuition in the setting up the nested models, they are rather cumbersome to work with. Inspired by \citet{Fithian:2015uj}, we will instead consider another sequence of random hypothesis that are more closely related to the nest models,
\[H_{0j}: \theta \in \mathcal{M}_j\left(X\right),\]
or equivalently, that $\theta_{[1]}, \ldots, \theta_{[j]}$ are not the best $j$ parameters in order.

Adapting this notation, the FWER can be viewed as $\PP\left[\text{reject } H_{0(j_0+1)}\right]$.

\paragraph{Special case of the BasicStop procedure}

While impractical, Procedure 3' can be thought of as performing all $n$ tests first, producing a sequence of $p$-values $p_j$, and returning
\begin{equation}
\label{eq:basicstop}
\hat{j}_0 = \min\left\{j: p_j > \alpha\right\} - 1.
\end{equation}
This is a special case of the BasicStop procedure. Instead of simply checking that Procedure 3' fits all the requirement for FWER control in BasicStop, we will give the construction of Procedure 3', assuming that we are to estimate $j_0$ with BasicStop.

In general, the FWER for BasicStop can be rewritten as $\PP\left[p_{j_0+1} \le \alpha\right]$. This is however difficult to analyze, as $j_0$ itself is random and dependent on $X$, thus we break the FWER down as follows:
\begin{align*}
\PP\left[p_{j_0+1} \le \alpha\right] & = \sum_j \PP\left[p_{j_0+1} \le \alpha \;\middle|\; j_0 = j\right] \PP\left[j_0 = j\right] \\
& = \sum_j \PP\left[p_{j+1} \le \alpha \;\middle|\; j_0 = j\right] \PP\left[j_0 = j\right] \\
& = \sum_j \PP\left[p_{j+1} \le \alpha \;\middle|\; \theta \in \mathcal{M}_{j+1}\left(X\right) \setminus \mathcal{M}_j\left(X\right)\right] \PP\left[j_0 = j\right].
\end{align*}
We emphasize here that $\theta$ is {\em not} random, but $\mathcal{M}_{j+1}$ is. Thus it suffices to construct the $p$-values such that
\begin{equation}
\label{eq:basicstopreq}
\PP\left[p_j \le \alpha \;\middle|\; \theta \in \mathcal{M}_j\left(X\right) \setminus \mathcal{M}_{j-1}\left(X\right)\right] \le \alpha \quad \text{for all } j.
\end{equation}

\paragraph{Considerations for conditioning}

By smoothing, we are free to condition on additional variables in \eqref{eq:basicstopreq}. A logical choice that simplified \eqref{eq:basicstopreq} is conditioning on the variables $\mathcal{M}_{j-1}\left(X\right)$ and $\mathcal{M}_j\left(X\right)$. Note that the choice of the model $\mathcal{M}_j\left(X\right)$, once again, based solely on the random indices $[1], \ldots, [j]$, so conditioning on both $\mathcal{M}_{j-1}\left(X\right)$ and $\mathcal{M}_j\left(X\right)$ is equivalent to conditioning on the random indices $[1], \ldots, [j]$, which in turns is equivalent to conditioning on the $\sigma$-field generated by the partition of the observation space $X$
\[\left\{\left\{X_{\tau_1} \ge \cdots \ge X_{\tau_j} \ge \max_{k>j} X_{\tau_k}\right\}:\; \tau \text{ is any permutation of } \left(1, \ldots, n\right)\right\},\]
or colloquially, the set of all possible choices of $[1], \ldots, [j]$. Within each set in this partition, the event $\left\{\theta \in \mathcal{M}_j\left(X\right) \setminus \mathcal{M}_{j-1}\left(X\right)\right\}$ is simply $\left\{\theta_{\tau_1} > \cdots > \theta_{\tau_j} \text { and } \theta_{\tau_j} \le \max_{k>j} \theta_{\tau_k}\right\}$, a trivial event.

As a brief summary, we want to construct $p$-values $p_j$ such that
\[\PP_{\substack{\theta_{\tau_1} > \cdots > \theta_{\tau_j} \\ \theta_{\tau_j} \le \max_{k>j} \theta_{\tau_k}}} \left[p_j \le \alpha \;\middle|\; X_{\tau_1} \ge \cdots \ge X_{\tau_j} \ge \max_{k>j} X_{\tau_k}\right].\]

\paragraph{Construction of the $p$-values}

To avoid the clutter in the subscripts, we will drop the $\tau$ in the subscript. Hence our goal is now
\[\PP_{\substack{\theta_1 > \cdots > \theta_j \\ \theta_j \le \max_{k>j} \theta_k}} \left[p_j \le \alpha \;\middle|\; X_1 \ge \cdots \ge X_j \ge \max_{k>j} X_k\right]\]
Construction of $p_j$ for other permutations $\tau$ can be obtained similarly.

There are many valid options for $p_j$ (such as constant $\alpha$). We will follow the idea in the proof of Part 1 of \Cref{thm:main} here. $p_j$ is intended to test $H_{0j}: \theta \in \mathcal{M}_j\left(X\right)$, which is equivalent to the union of the null hypotheses:
\begin{enumerate}
\item $\theta_k \le \theta_{k+1}$ for $k = 1, \ldots, j-1$, and
\item $\theta_j \le \theta_k$ for $k = j+1, \ldots, n$. (The union of these null hypotheses is $\widetilde{H}_{0j}$.)
\end{enumerate}

Since the joint distribution of $X$, restricted to $\left\{X_1 \ge \cdots \ge X_j \ge \max_{k>j} X_k\right\}$, remains in the exponential family, we can construct the $p$-values for each of the hypotheses above by conditioning on the variables corresponding to the nuisance parameters here, similar to the proof of Part 1 of \Cref{thm:main}. Then we can take $p_j$ as the maximum of such $p$-values.

For the hypothesis $H_{0jk}: \theta_j \le \theta_k$, we can construct $p_{jk}$, by considering the survival function of the conditional law
\begin{align*}
&~ \mathcal{L}_{\theta_j = \theta_k} \left(D_{jk} \;\middle|\; \left\{X_1 \ge \cdots \ge X_j \ge \max_{\ell>j} X_\ell\right\}, X_{\setminus\left\{j, k\right\}}, M_{jk}\right) \\
= &~ \mathcal{L}_{\theta_j = \theta_k} \left(D_{jk} \;\middle|\; \left\{X_{j-1} \ge X_j \ge \max_{\substack{\ell>j \\ \ell \ne k}} X_\ell \text{ and } X_j \ge M_{jk}\right\}, X_{\setminus\left\{j, k\right\}}, M_{jk}\right)
\end{align*}

Once again, $X_{j+1} = \max_{\ell > j} X_\ell$ is simply shorthand for simplifying our notation. Now the $p$-values are similar to the ones in \Cref{eq:p1j}, for $k>j$:

\[p_{jk} = \frac{\int_{D_{jk}}^{X_{j-1}} g\left(X_1, \ldots, M_{jk} + z, \ldots, M_{jk} - z, \ldots, X_n\right) \,dz}{\int_{\max\left\{X_{j+1} - M_{jk}, 0\right\}}^{X_{j-1}} g\left(X_1, \ldots, M_{jk} + z, \ldots, M_{jk} - z, \ldots, X_n\right) \,dz}.\]

We can graphically represent $p_{jk}$ in \Cref{fig:croprays}, a diagram analogous to \Cref{fig:comparerays}.

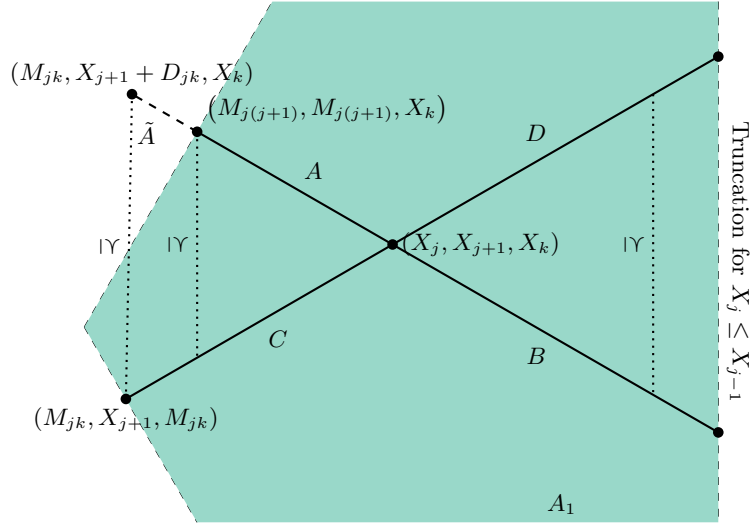
\begin{figure}[htbp]
\centering
\begin{tikzpicture}
	\coordinate (C1) at (60:5);
	\coordinate (C2) at (300:3);
	\coordinate (A) at (60:3);
	\coordinate (B) at ($(A) + (330:3)$);
	\coordinate (C) at ($(0, 0)!(B)!(C2)$);
	\coordinate (D) at ($(A) + (150:1)$);
	\coordinate (C3) at ($(B) + (330:5)$);
	\coordinate (C4) at ($(B) + (30:5)$);
	\coordinate (C5) at ($(C3)!(C1)!(C4)$);
	\coordinate (C6) at ($(C3)!(C2)!(C4)$);
	\draw[dashed] (0, 0) -- (C1);
	\draw[dashed] (0, 0) -- (C2);
	\draw[dashed] (C5) -- (C6) node[midway, above, rotate = -90] {Truncation for $X_j \le X_{j-1}$};
	\fill[mid_green] (0, 0) -- (C1) -- (C5) -- (C6) -- (C2) -- cycle;
	\node[above] at ($(C2)!0.7!(C6)$) {$A_1$};
	\draw[thick] (A) -- (B) node[midway, above right] {$A$};
	\draw[thick] (B) -- (C3) node[midway, below left] {$B$};
	\fill (C3) circle(2pt);
	\fill (A) circle(2pt) node[above right] {$\left(M_{j\left(j+1\right)}, M_{j\left(j+1\right)}, X_k\right)$};
	\fill (B) circle(2pt) node[right] {$\left(X_j, X_{j+1}, X_k\right)$};
	\draw[thick] (C) -- (B) node[midway, below right] {$C$};
	\node[below] at (C) {$\left(M_{jk}, X_{j+1}, M_{jk}\right)$};
	\draw[thick] (B) -- (C4) node[midway, above left] {$D$};
	\fill (C4) circle(2pt);
	\fill ($(0, 0)!(B)!(C2)$) circle(2pt);
	\draw[thick, dashed] (A) -- (D) node[midway, below left] {$\tilde{A}$};
	\fill (D) circle(2pt) node[above] {$\left(M_{jk}, X_{j+1} + D_{jk}, X_k\right)$};
	\draw[thick, dotted] (C) -- (D) node[midway, below, rotate=-90] {$\succeq$};
	\draw[thick, dotted] (A) -- ++(0, -3) node[midway, below, rotate=-90] {$\succeq$};
	\draw[thick, dotted] ($(B) + (30:4)$) -- ($(B) + (330:4)$) node[midway, below, rotate=-90] {$\succeq$};
\end{tikzpicture}
\caption{The two $p$-values constructed corresponds to taking integrals of $g$ along these segments, that lie on a level set of $x_j + x_{j+1} + x_k$. The dashed line corresponds to extension in \eqref{eq:int_extension}. The dotted line on the far right is the truncation that enforces $X_j < X_{j-1}$.}
\label{fig:croprays}
\end{figure}

We have $p_{j(j+1)} \ge \max_{k>j} p_{jk}$ by \Cref{sec:winnerbest}: the upper truncation for $X_j$ can be represented by cropping \Cref{fig:comparerays} along a vertical line, shown in \Cref{fig:croprays}. Considering $p_{j(j+1)}$ is sufficient in rejecting all the $H_{0jk}$. We will take $p_{j*} = p_{j(j+1)}$, noting that this is the $p$-value that Procedure 3' would produce. In fact, $p_{j*}$ is also the $p$-value we would have constructed if we were to reject only $\widetilde{H}_{0j}$.

Upon constructing $p_j$, one should realize that the $p$-values for testing $\theta_k \le \theta_{k+1}$ would have been constructed in earlier iterations of BasicStop, as $p_{k*}$. In other words, $p_j = \max_{k \le j} p_{k*}$ is the sequence of $p$-values that works with BasicStop. However, from \eqref{eq:basicstop},
\[\hat{j}_0 = \min\left\{j:\; \max_{k \le j} p_{k*} > \alpha\right\} - 1 = \min\left\{j:\; p_{j*} > \alpha\right\} - 1,\]
so it is safe to apply BasicStop to $p_{j*}$ directly, yielding Procedure 3'.
\end{proof}

\begin{customthm}[Part 3 of \Cref{thm:main}]
Assume the model~\eqref{eq:expfam} holds and $g\left(x\right)$ is a Schur-concave function. Procedure 3 is a conservative stepwise procedure with FWER no larger than $\alpha$.
\end{customthm}

\begin{proof}
The $p$-values $p_{j\left(j+1\right)}$ obtained in Procedure 3' are always smaller than their counterpart in Procedure 3, as the upper truncation at $X_{j-1}$ is on the upper tail. Therefore Procedure 3 is conservative and definitely valid.
\end{proof}

\section{Discussion}
\label{sec:disc}

Combining ideas from conditional inference and multiple testing, we have proven the validity of several very simple and seemingly ``naive'' procedures for significance testing of sample ranks. In particular, we have shown that an unadjusted pairwise test comparing the winner with the runner-up is a valid significance test for the first rank. Our result complements and extends pre-exisiting analogous results for location and location-scale families with independence between observations. Our approach is considerably more powerful than previously known solutions. We provide similarly straightforward conservative methods for producing a lower confidence bound for the difference between the winner and runner up, and for verifying ranks beyond the first.

Claims reporting the ``winner'' are commonly made in the scientific literature, usually with no significance level reported or an incorrect method applied. For example, \citet{Uhls:2012gf} asked $n = 20$ elementary and middle school students which of seven personal values they most hoped to embody as adults, with ``Fame'' ($8$ responses) being the most commonly selected, with ``Benevolence'' ($5$ responses) second. The authors' main finding --- which appeared in the abstract, the first paragraph of the article, and later a CNN.com headline \citep{Alikhani:2011} --- was that ``Fame'' was the most likely response, accompanied by a significance level of $0.006$, which the authors computed by testing whether the probability of selecting ``Fame'' was larger than $1 / 7$. The obvious error in the authors' reasoning could have been avoided if they had performed an equally straightforward two-tailed binomial test of ``Fame'' vs.\ ``Benevolence,'' which would have produced a $p$-value of $0.58$.

\section*{Reproducibility}
A git repository containing with the code generating the image in this paper is available at \url{https://github.com/kenhungkk/verifying-winner}.

\Urlmuskip=0mu plus 1mu\relax
\bibliographystyle{imsart-nameyear}
\bibliography{papers}

\end{document}